\documentclass[10pt]{article}
\usepackage{amssymb}
\usepackage{amsmath}
\usepackage[all]{xy}
\usepackage{amsmath}
\usepackage{amsfonts}
\usepackage{amssymb}
\usepackage{amscd}
\usepackage{amsthm}
\usepackage{latexsym}
\usepackage{amsbsy}
\usepackage{graphicx}
\usepackage{float}
\usepackage[toc,page]{appendix}
\usepackage{caption}
\usepackage{subcaption}
\usepackage{soul}

\newtheorem{prop}{Proposition}
\DeclareMathOperator{\tr}{Tr}

\makeatletter
\def\section{\@startsection{section}{1}{\z@}{-3.25ex plus -1ex minus
		-.2ex}{1.5ex plus .2ex}{\normalfont\bfseries}}
\def\subsection{\@startsection{subsection}{1}{\z@}{-3.25ex plus -1ex
		minus -.2ex}{1.5ex plus .2ex}{\normalfont\itshape}}
\date{}
\makeatother
\title{ Phase Transition in Random Noncommutative Geometries}
\author{ Masoud Khalkhali and  Nathan Pagliaroli\\
Department of Mathematics, University of Western Ontario\\
London, Ontario, Canada\footnote{\emph{Email addresses}:  masoud@uwo.ca, npagliar@uwo.ca}}

\begin{document}
	\maketitle
	\begin{abstract}
	We present an analytic  proof of the existence of phase transition in the large $N$ limit of certain random noncommutaitve geometries. These geometries can be expressed as  ensembles of Dirac operators. When they  reduce to single matrix ensembles, one can apply the Coulomb gas method to find the empirical spectral distribution.  We elaborate on the nature of  the large $N$ spectral distribution of the Dirac operator itself. Furthermore, we  show that these models exhibit both a single and double cut region for certain values of the order parameter and find the exact value where the transition occurs. 
	\end{abstract}
	\section{Introduction}
	In noncommutative geometry, the notion  of  {\it spectral triple} simultaneously encodes  the data of a  Riemannian metric,  a spin structure, and a  Dirac operator on  a noncommutative space   \cite{Connes95}. The space can be discrete or continuous, finite or infinite dimensional,  and the possibility of doing geometry and topology over noncommutative spaces equipped with a spectral triple is certainly an attractive idea. In fact allowing noncommutativity is essential in order to deal with spaces such as {\it bad quotients} of standard spaces. In this version of geometry, {\it Dirac operators}, suitably abstracted and understood, take the centre  stage. For example,  at the  most basic level, we have  the distance formula of Connes  \cite{Connes94}
	$$ d(p, q) = \text{Sup}\{|f(p)-f(q)|; \, ||[D, f]|| \leq 1\},$$
which shows that the geodesic distance on a Riemannian (spin) manifold  can be 	recovered from the action of its Dirac operator $D$ on the $L^2$ space of spinors.  At a more fundamental level,  the   {\it reconstruction theorem}  of  Connes \cite{Connes-2013}  shows that   a spin Riemannian manifold can be fully recovered if we are given a commutative spectral triple satisfying some natural conditions. 
	  Thus one  can consider a   spectral triple  as a noncommutative spin Riemannian manifold, and the Dirac operator as an avatar of the Riemannian metric. In particular the space of Riemannian metrics on a fixed space can be equally thought of as the space of Dirac operators.

Motivated by  this  idea, Barrett and Glaser proposed  a toy model of Euclidean quantum gravity  as an ensemble of Dirac operators on a fixed {\it finite} noncommutative background space \cite{Barrett2015, Barrett2016}. In this paper we simply think of this model as a {\it random noncommutative geometry}. The {\it partition function} of these ensembles are   of the form 
\begin{equation}\label{geometries}
\mathcal{Z} = \int e^{- \mathcal{S} (D)}  d D \, ,
\end{equation}
where the action functional ${\mathcal{S} (D)}$ is defined in terms of the spectrum of the Dirac operator $D$, and the integral is over the moduli space of Dirac operators compatible with a  fixed  {\it finite noncommutative geometry}. In particular Dirac operators are dynamical variables and play the role of metric  fields and   the moduli of Dirac operators is typically a finite dimensional vector space. It should be stressed that the action functional $\mathcal{S}$ is chosen in such a way that the partition function \eqref{geometries} is absolutely convergent and finite. A good choice for exampe would be $\mathcal{S}(D)= \text{Tr} (f (D))$ for a  real polynomial $f$ of even degree with positive leading coefficient. 

Using the classifications of finite spectral triples  and their Dirac operators  one can express these models as {\it multi-trace} and {\it multi-matrix}  random matrix models 
\cite{Barrett2015, Barrett2016} (cf. also  \cite{Shahab thesis, AK}). This shows that   the analytic study of these models 
as convergent matrix integrals is quite hard in general. In \cite{Barrett2016}, computer study of these models is carried via Monte Carlo Markov chain methods. They indicate phase transition and multi-cut regimes in their spectral distribution. 
In \cite{Shahab thesis, AK},  formal aspects of these models and their generalizations  is studied through topological recursion techniques.  In \cite{Sanchez}  the algebraic  structure of the action functional of these models is further analyzed.

In this paper   we look at the simplest types of these models and derive  explicit analytic expressions for their  large $N$ limits  of empirical spectral distributions.    
In the classification scheme of \cite{Barrett2015}, the underlying finite  spectral triples we consider are  of  types $(1, 0)$ and $(0, 1)$. Using  the Coulomb gas method outlined in \cite{Deift}, 
we rigorously prove the existence of a {\it phase transition}  between a single and double cut phase of the proposed models by finding the {\it equilibrium measure} in both cases. The transition happens at the critical value of $-5\sqrt{2}/2$ for both models. Furthermore, both models have the same global 
behavior of eigenvalues in the large $N$ limit. With explicit formulas for the empirical spectral distributions of these models,  we then elaborate  on  the large $N$ limiting distribution of the Dirac operators.

The Dirac operator of type ${(1,0)}$  matrix geometries is given by an anticommutator \cite{Barrett2016}:
\begin{equation*}
D = \{ H , \cdot \} \, , \quad H \in \mathcal{H}_N \, ,
\end{equation*}
and for   $(0, 1)$ matrix geometries  by a commutator 
\begin{equation*}
D = -i\,[ L , \cdot ] \, , \quad L \in \mathcal{L}_N \, ,
\end{equation*}
where ${\mathcal{H}_N}$ and ${\mathcal{L}_N}$ denote the space of ${N \times N}$ Hermitian and   anti-Hermitian matrices,  respectively. The Dirac operators for these geometries  act on the Hilbert space $ \mathrm{M}_N (\mathbb{C})$, and 
\begin{equation*}
d D = d H = \prod_{i=1}^N d H_{ii} \, \prod_{1 \leq i < j \leq N}  d ({\text{Re}} (H_{ij})) \, d ({\text{Im}} (H_{ij})) \, 
\end{equation*}
is the canonical Lebesgue measure on ${\mathcal{H}_N}$. In this paper we consider a quartic action
\begin{equation*}
	S(D) = g\tr D^2 + \tr D^4,
	\end{equation*}
where $g$ is a coupling constant. 
This action can be expressed  in terms of Hermitian and anti-Hermitian matrices for $(1, 0)$ and $(0, 1)$ geometries, respectively.   The results are as follows:  
	\begin{equation} \label{bt1}
	S^{(1,0)}(D) = 2N(g \tr H^2 +   \tr H^4) + 2g (\tr H)^2 + 8 \,\tr H \tr H^3 + 6(\tr H^2)^2 ,
	\end{equation}
	\begin{equation*}
	S^{(0,1)}(D) = 2N(-g\tr L^2 +  \, \tr L^4) + 2g (\tr L)^2 - 8 \,\tr L \tr L^3 + 6(\tr L^2)^2.
	\end{equation*}
	
	Furthermore, one may write the $(0,1)$ model in terms of Hermitian matrices as well.   This gives us
	
	\begin{equation} \label{bt2}
	S^{(0,1)}(D) = 2N(g\tr H^2 +  \, \tr H^4) - 2g (\tr H)^2 - 8 \,\tr H \tr H^3 + 6(\tr H^2)^2.
	\end{equation}
	We can thus treat both models as  a bitracial matrix model. The particular structure of the $H$-action is solely dictated by the quartic $D$-action.

This paper is organized as follows. In Section two we derive the saddle point equations for models  of type $(1, 0)$ and $(0, 1)$ and write down the action functional for their equilibrium measures. The corresponding Euler-Lagrange equations  are in fact integral equations, whose solutions correctly recover the sought after equilibrium measures. In Section three we present the equilibrium measures of both models which turns out to be identical, in both the single and double cut cases. We discuss their behaviour and the behaviour of limiting measure of the Dirac operators  around the phase transition. Detailed calculations are given in the Appendix in Section five. Our conclusions and outlook are presented in Section four, where we suggest a few possible directions one may approach more complicated multi-matrix models associated with these spaces.

\section{The Saddle Point Equation}

	As we pointed out in the introduction, the $(1, 0)$ and $(0, 1)$ random noncommutative geometries 
	can be  described as  bitracial Hermitian matrix models 	defined by \eqref{bt1} and \eqref{bt2}, respectively. 
		Now our partition functions are  of the form 
		\begin{equation} \label{matrix integral}
			\mathcal{Z} = \int_{ \mathcal{H}_{N}} e^{F(H)} d H,
		\end{equation}
		where 
$F$ is a unitary invariant function.
More concretely,  the $(1, 0)$ and $(0, 1)$ models,  when passed to eigenvalue space using Weyl's integration formula \cite{Random Matrices} are of the form
\begin{equation} \label{jpd}
 	\mathcal{Z} = C_{N} \int_{\mathbb{R}^{N}} e^{-N\sum_{i=1}^{N}V(\lambda_{i}) - \sum_{i,j=1}^{N}U(\lambda_{i}, \lambda_{j})} \prod_{1\leq i  < j \leq N}(\lambda_{i}-\lambda_{j})^{2}d\lambda_{1}...\lambda_{N},
\end{equation}
where $C_N$ is a constant whose exact value is not important for this paper. The polynomial $V$ comes from the single trace contribution in the potential, and $U$ the bi-tracial contribution. In the 
$ (1, 0)$ case passing equation (\ref{bt1}) to eigenvalues we have
\begin{equation*}
	V(s) = 2g s^2+ 2s^4,
\end{equation*}
\begin{equation*}
	U(s,t) = 2 gst  + 8 st^{3} + 6 s^{2}t^{2},
\end{equation*}
and for the $(0, 1)$ case passing equation (\ref{bt1}) to eigenvalues
\begin{equation*}
V(s) = 2g s^2+ 2s^4,
\end{equation*}
\begin{equation*}
U(s,t) = -2 gst  - 8 st^{3} + 6 s^{2}t^{2},
\end{equation*}
for a coupling constant $g$. As we will  see later,  by symmetry,  the terms  with odd number of $s$ or $t$  will  have no contribution to the limiting spectral distribution. Consider the limit 
\begin{equation*}
	\lim_{N\rightarrow \infty} \frac{1}{N^{2}}\ln \left[\int_{\mathbb{R}^{N}} e^{-N\sum_{i=1}^{N}V(\lambda_{i}) - \sum_{i,j=1}^{N}U(\lambda_{i},\lambda_{j})} \prod_{1\leq i  < j \leq N}(\lambda_{i}-\lambda_{j})^{2}d\lambda_{1}...\lambda_{N}\right].
\end{equation*}
Heuristically, one expects that the leading contribution to the integral comes from places where the integrand achieves its maximum using Laplace's method.  Such a maximizing point $\{\lambda^{*}_{1},...,\lambda^{*}_{N}\}$ can be found by setting the partial derivatives equal to zero, from which one can construct the normalized counting measure \text{expressed in terms of delta functions}

\begin{equation} \label{empirical}
	\alpha_{N} =\frac{1}{N} \sum_{i=1}^{N}\delta_{\lambda_{i}^{*}}.
\end{equation}

It is not hard to see from \cite{Deift} that there indeed exists a point in $\mathbb{R^{N}}$ that maximizes the integrand, hence $\alpha_{N}$ exists. Furthermore, it converges in the vague topology to a unique probability measure that has compact support. The limit of  $\alpha_{N}$, denote it $\mu^{Q}$, is referred to as the equilibrium measure and it is the Borel probability measure that minimizes the energy functional 

\begin{align}\label {Vector equilibrium}
&I ^ {Q} (\mu): =\int V ( {s}) d\mu ( {s}) +\int\int U ( {s}, {t}) d\mu ( {s}) d\mu ( {t}) -\int\int \textrm {log} | s - t | d\mu (s) d\mu (t).
\end{align}
 
Now assume that $\mu^{Q} = \Psi(x)dx$ where $\Psi(x)$ is continuous. If we can show that such a measure satisfies our conditions, then by uniqueness it must be the equilibrium measure. From the Euler-Lagrange equations derived from this problem we get the corresponding saddle point equation
\begin{align}\label{saddle point eqn}
\begin{split}
	\text{P.V.}\int_{\text{supp} \Psi} \frac{\Psi(s)}{s-x}ds &= \frac{1}{2}\left[4gx + 8x^{3}+ \frac{\partial}{\partial x}\int (\pm 2gxy\pm 8xy^{3} + 6x^{2}y^{2} )\Psi(y)dy\right]\\
&= 2gx + 4x^{3}+ \int (\pm gy \pm 4y^{3} + 6xy^{2} )\Psi(y)dy\\
&= 2gx + 4 x^{3} \pm g m_{1} \pm 4 m_{3} + 6 m_2 x. 
\end{split}
\end{align}
Here, the plus sign refers to the $(1, 0)$ model, the minus sign refers to the $(0, 1)$ model, and $m_{n}$ denotes the $n$th moment of the distribution. We now consider the two possible cases that the  $\Psi$  support has one or   two cuts. 

\section{Results}
\subsection{Symmetric 1-cut region}
Assume that $\text{supp} \, \Psi =[-2a, 2a]$ for some $a>0$. In general, the integral equation
\begin {equation*}
\int_{-2 a} ^ {2 a} \frac {\Psi (y) d y} {y - x} = A x + B x ^ {3} + C
\end {equation*}
has the solution
\begin {equation*}
\Psi (x) =\frac {1} {\pi} (A +2B a ^ {2} + Bx ^ {2})\sqrt {4 a ^ {2} - x ^ {2}}_{+}.
\end {equation*}
Based on the limiting behavior of the Stieltjes transform it can be deduced that
\begin{equation}\label{limiting Plemelj behaviour}
	\frac{1}{a}= 2Aa + 6Ba^{3}.
\end{equation}
Additionally the first few moments are
\begin{align*}
	m_{1}&= 0\\
	m_{2}&=2Ba^{6}  + a^{2}\\
	m_{3}&= 0.\\ 
\end{align*}
For more details see appendix A.1. The odd moments are equal to zero by symmetry. In both $(0, 1)$ and $(1, 0)$ models
\begin{align*}
m_{2}&=8 a^{6} + a^{2},
\end{align*}
and equation (\ref{limiting Plemelj behaviour}) gives us

\begin{align} \label{1,0 poly}
\begin{split}
g&= -24a^{6} -9a^{2} +\frac{1}{4a^{2}}.
\end{split}
\end{align}

 One could also find the roots of the degree eight polynomial (which can be reduced to a degree 4 by a change of variables) to get an expression for the limiting distribution strictly in terms of $g$.  However, the actual expression for the roots in terms of $g$ is unsightly so it will not be included in this document. It can easily be found with any computer algebra software such as $\text{\it maple }$ or $\text{\it mathematica}$ and there are in fact only two unique real roots that differ only by sign. The spectral density function is
   
\begin{align*}
\Psi(x)&=  \frac{1}{\pi}(-4a^{2} +\frac{1}{2a^{2}} + 4x^{2})\sqrt{4a^{2}-x^{2}}_{+}.\\	
\end{align*}

When we numerically graph this distribution we find that at some critical value of $g$ it dips below zero and splits into a two-cut case. See the below figure.
\begin{figure}[H]\label{DOS single cut}
	\centering
	\includegraphics[width=0.5\textwidth]{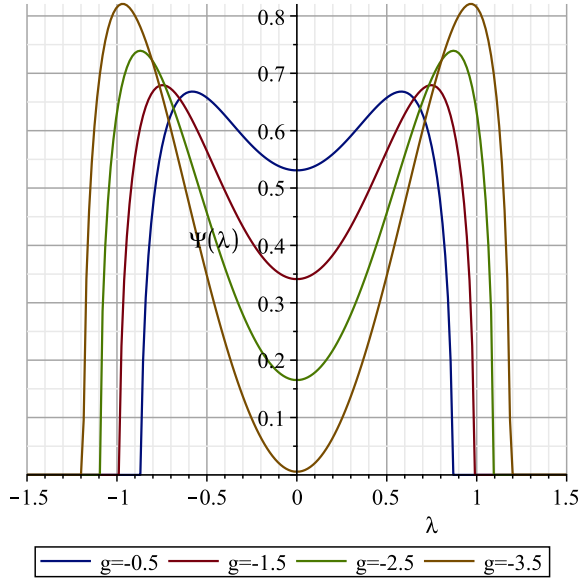}
	\caption{The equilibrium measure for the $(1, 0)$ and $(0, 1)$ models from the single cut analysis.}
\end{figure}
A precise critical value is found by setting $\Psi(x) =0$ and $x=0$ and isolating for $a$, giving us $a_{c} = \frac{1}{\sqrt[4]{8}}$. Plugging $a_{c}$ into (\ref{1,0 poly}) we find
\begin{align*}
	g_{c} = -\frac{5\sqrt{2}}{2} \approx -3.5,
\end{align*}
which is the value found by Barrett and Glaser in  \cite{Barrett2016} by computer simulation  as  evidence for a phase transition in the $(1, 0)$ model. However, they  did not find any evidence of a phase transition in the $(0, 1)$ case. We are unsure of the cause of this discrepancy, but  it may be due to  the small matrix size (10 by 10). We emphasize that results seen in this paper are \textbf{exclusively in the large $N$ limit}, which may exhibit  different phenomenon than finite $N$. We were genuinely surprised that the critical point in the limit of the (1,0) model was so close to that found for finite matrix size in \cite{Barrett2016}. In the large $N$ limit several of the terms  in $U$ vanish due to asymmetry. The lack of contribution from these terms in the limit may be the cause for discrepancy in the behaviour of these models.
\subsection{Symmetric 2-cut region}
It would appear that we now have to change our analysis to the double-cut case for $g<g_{c}$. Assume $\text{supp} \, \Psi =[-a,-b]\cup [b,a]$. Via the resolvent method, our spectral density function is of the form
\begin{align*}
	\begin{split}
	\Psi(x) = \frac{2}{\pi}|x| \sqrt{(x^{2}-a^{2})(b^{2}-x^{2})}_{+},
	\end{split}
\end{align*}
where the moments are
\begin{align*}
	m_{1}&=0\\
	m_{2}&= -\frac{1}{5}g\\
	m_{3}&= 0.
\end{align*}
We have the additional conditions that 

\begin{equation}\label{a in terms of g}
	a^{2}= -\frac{1}{5}g +\frac{\sqrt{2}}{2},
\end{equation}
and 
\begin{equation}\label{b in terms of g}
b^{2}= -\frac{1}{5}g -\frac{\sqrt{2}}{2}.
\end{equation}
For more details see A.2.

\begin{figure}[H]\label{DOS of (1,0) multicut}
	\centering
	\includegraphics[width=0.5\textwidth]{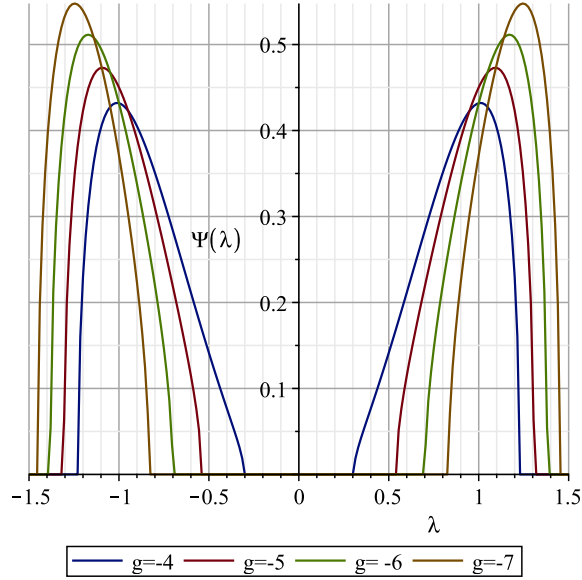}
	\caption{The equilibrium measure  for $(1, 0)$ and $(0, 1)$ from the double cut analysis.}
\end{figure}





\subsection{Spectrum of the Dirac operators}
Let $\Psi_{N} (\lambda_{1}, ...,\lambda_{N})$ be the joint probability distribution from equation \eqref{jpd}. Recall the $l$-th marginal distribution density defined as 
\begin{equation*}
	\Psi_{N}^{l}(\lambda_{1},...,\lambda_{l}) = \int_{\mathbb{R}^{N-l}} \Psi_{N} (\lambda_{1}, ...,\lambda_{N}) d\lambda_{l+1}...d\lambda_{N}.
\end{equation*}
It is well known that for invariant ensembles where the $U$ interaction term is absent the 1-th marginal distribution converges to the equilibrium measure \cite{Stat mech}
\begin{equation*}
	\lim_{N\rightarrow \infty} \Psi^{1}_{N}(x) = \Psi(x).
\end{equation*}
Furthermore from \cite{Universality}, for any distinct eigenvalues there is a weak factorization of $\Psi_{2}^{N} (\lambda_{1}, \lambda_{2})$ into the product $\Psi^{N}_{1}(\lambda_{1})\Psi^{N}_{1}(\lambda_{2})$, with a correction term. Additionally, this correction term vanishes in the large $N$ limit, giving a product distribution
\begin{equation*}
	\Psi_{2}(\lambda_{1},\lambda_{2}) = \Psi(\lambda_{1})\Psi(\lambda_{2}).
\end{equation*}
We will assume these results hold when $U$ is nonzero allowing us to write an expression for the spectral density of the Dirac operator, as done in \cite{Barrett2016}. Dirac operators are linear operators on a $N^{2}$ dimensional vector space so they can have up to $N^{2}$ distinct eigenvalues $\omega_{jk}$. These eigenvalues can in fact be written in terms of the eigenvalues $\lambda_{j}$ of the random matrix $H$ or the $i \lambda_{j}$ of $L$. In particular for the $ (1, 0)$ and $(0, 1)$ cases, they can be decomposed as 
\begin{equation*}
	\omega_{jk} = \lambda_{j}\pm \lambda_{k},
\end{equation*}
respectively. Consider the expectation value of an observable $f(\omega)$:
\begin{equation*}
	\langle f\rangle =\int\int_{\text{supp} \Psi}f(\lambda\pm \lambda') \Psi(\lambda)\Psi(\lambda')d\lambda d\lambda',
\end{equation*}
which can be rewritten as
\begin{equation*}
	=\int f(\lambda \pm \lambda') \left(\int \Psi(\lambda) \Psi(\lambda')d\lambda\right) d\lambda' = \int f(\omega) \left(\int \Psi(\omega \mp \lambda) \Psi(\lambda')d\lambda\right) d\omega = \int f(\omega) \rho(\omega) d\omega,
\end{equation*}
where 
\begin{equation*}
	\rho(\omega) := \int \Psi(\omega \mp \lambda) \Psi(\lambda)d\lambda
\end{equation*}
 is the spectral density of the Dirac operator. This convolution of $\Psi$ is an elliptic integral but can be graphed numerically for various values of $g$. Note that $\Psi$ is symmetric, so $\rho$ is the same for both models.

\begin{figure}[H]
	\centering
	\begin{subfigure}{0.5\textwidth}
		\centering
		\includegraphics[width=0.8\linewidth]{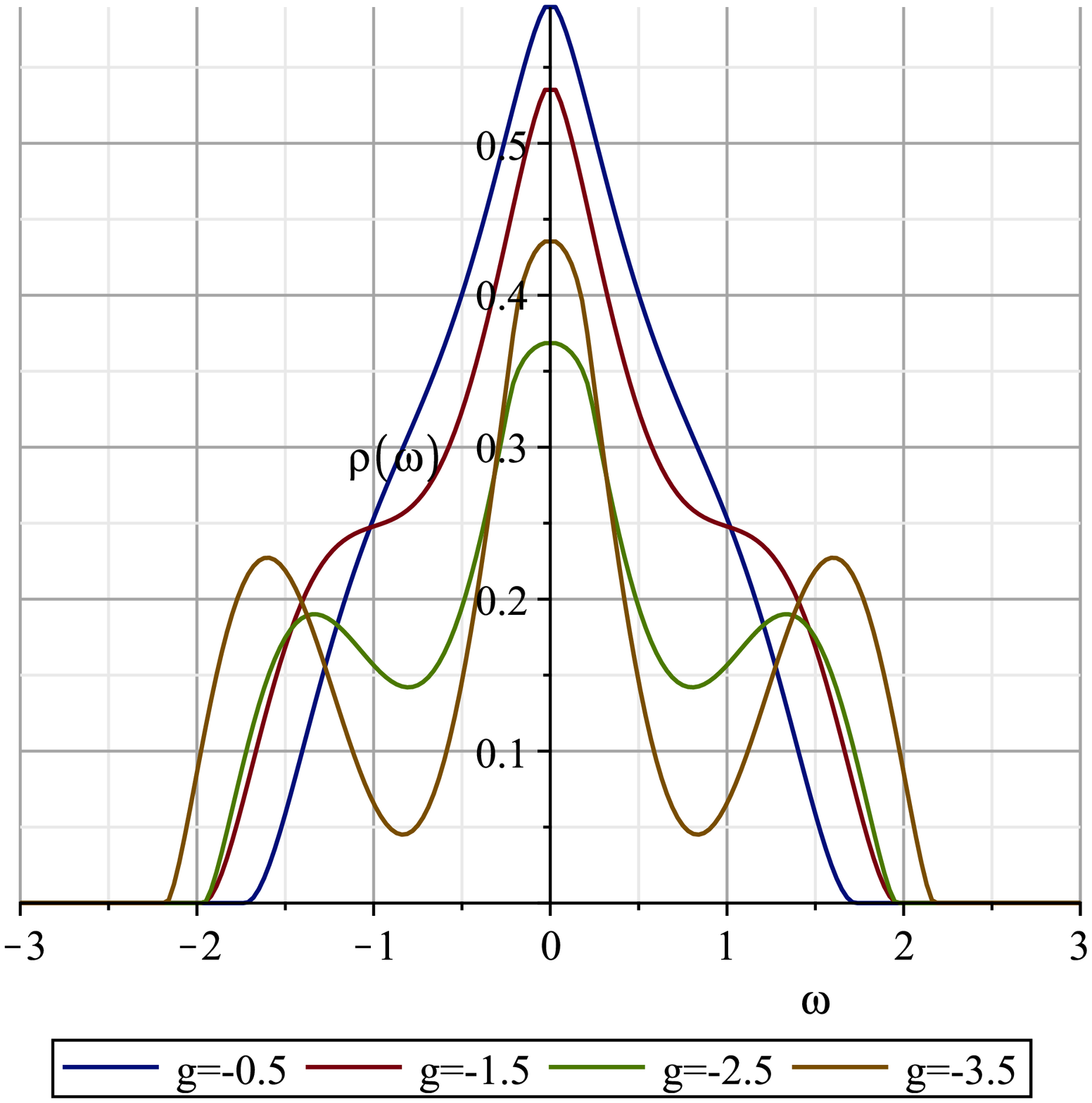}
		\label{fig:sub1}
	\end{subfigure}%
	\begin{subfigure}{0.5\textwidth}
		\centering
		\includegraphics[width=0.8\linewidth]{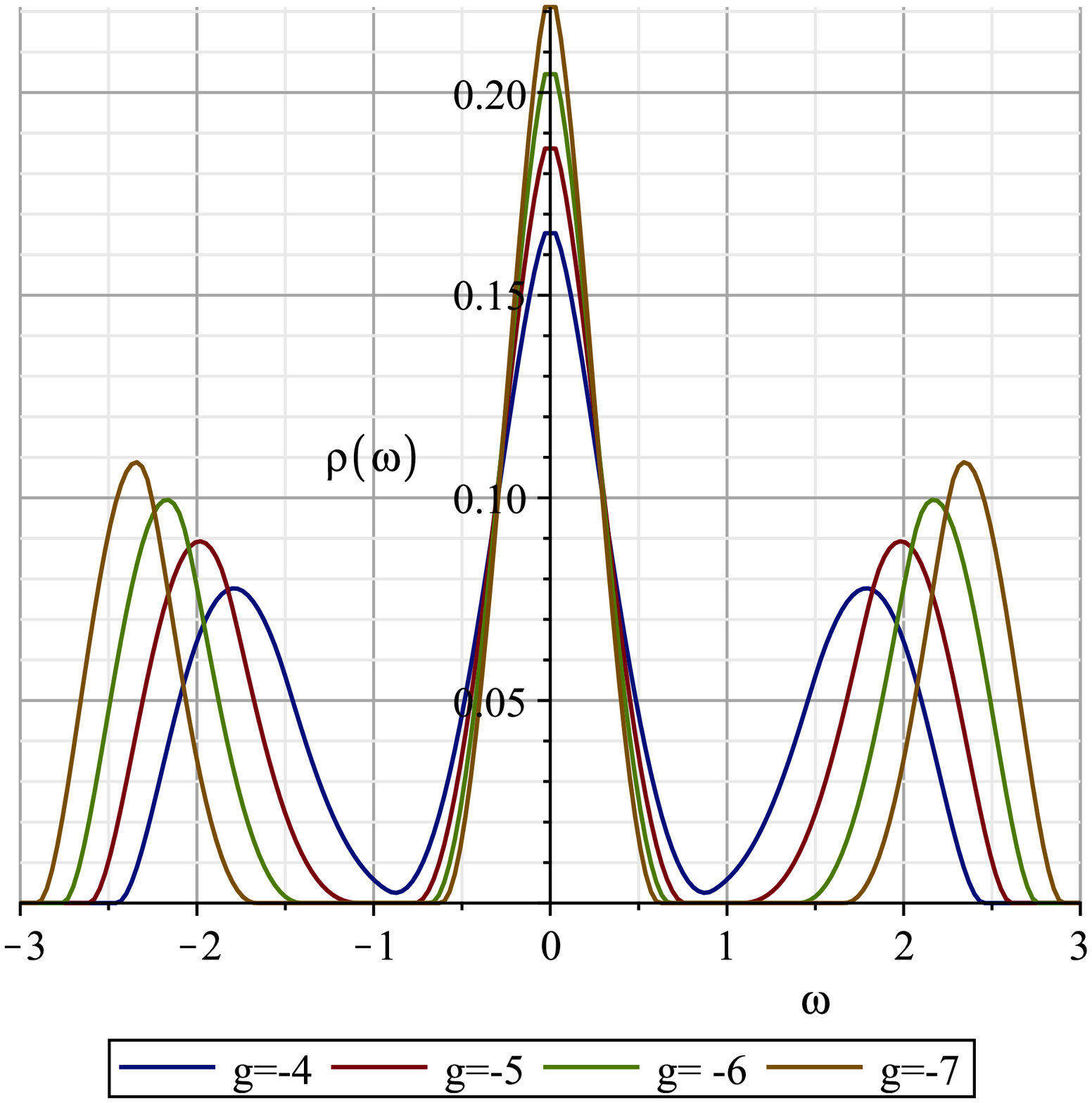}
		\label{fig:sub2}
	\end{subfigure}
	\caption{The density of states   for  type $(1, 0)$ and $(0, 1)$ Dirac operators.}
	\label{fig:test}
\end{figure}

\section{Conclusion}
In this paper we considered two random matrix models suggested by noncommutative geometric models of Barrett and Glaser \cite{Barrett2016}. We used analytic techniques to study these models and compared our results   with the results obtained by  computer simulations in \cite{Barrett2016}. Using equilibrium measure techniques we were able to explicitly compute the density of states in terms of the coupling constant $g$, and found that they were the same for both models. The precise critical value is found to be $g_{c} = -\frac{5\sqrt{2}}{2}$. In the process we explicitly computed the resolvent which is the same for both operators. One can also apply these techniques to $(1, 0)$ and $(0, 1)$ geometries with   higher order polynomial action functionals. It would also be interesting to compare these models to their corresponding formal matrix models using the techniques developed in \cite{Eynard2018}.

The more complicated matrix geometries are much harder to examine using classical analytic methods since they maintain no form of unitary invariance. Consider for example the type ${(2,0)}$ geometry from \cite{Barrett2016}. It contains the following term:
$$-8N\, \tr (H_{1}H_{2}H_{1}H_{2})$$
 which is far from being bi-unitary invariant and one cannot even apply the Harish-Chandra-Itzykson-Zuber formula \cite{planar II}. Models of a similar form have been approached using the method of characteristic expansion \cite{Migdal,Kazakov 1, Kazakov 2}. Another method to gain insight into these models might be through blobbed topological recursion \cite{blobbed, AK}.  Additionally one might be able to apply results from Free Probability theory as suggested in \cite{Sanchez}. Finally, in a different but related  direction, it would be interesting   to apply the random matrix theory techniques  we used in this paper to finite spectral triples based on {\it operator systems} developed recently in  \cite{Connes-VS}.

\section{Appendix}
In this appendix we present explicit calculations leading to the solution of the integral equation (\ref{saddle point eqn}). What one may notice is that these quartic single matrix multi-trace models can be handled in almost the exact same manner as the quartic case. The only difference is that the coefficients of the saddle point equation may contain the moments of the equilibrium measure. However, one must first establish the existence of such a measure to guarantee the existence of these moments. From there the techniques used are the same as for the quartic case.  In the single cut analysis we follow the traditional approach  introduced in \cite{Brezin1978} and for multi-cut case we follow the techniques used in \cite{Deift}. 
\subsection*{A.1 Single cut analysis}
\begin {prop} \label{Plemelj method}
Consider  the following integral equation
\begin {equation*}
	\text{P.V.}\int_{-2 a} ^ {2 a} \frac {\Psi (y) d y} {y - x} = A x + B x ^ {3} + C,
\end {equation*}
where $\Psi (y) $ is a continuous probability distribution defined on the interval $[-2a,2a]$ for a positive nonzero real number $a$, and where A, B, C are all real constants. Then
\begin {equation*}
\Psi (x) =\frac {1} {\pi} (A +2B a ^ {2} + Bx ^ {2})\sqrt {4 a ^ {2} - x ^ {2}}_{+}. 
\end {equation*}
\end {prop}
\noindent This follows from the Plemelj formula together with the Zhukovsky transform.

Now the value of $a$ can be found by considering the limiting behavior of the distribution, which we will show in the following section. Recall the limiting behavior of the Stieltjes transform of $\psi$ as $x \rightarrow \infty$ 
\begin{equation} \label{Hilbert transform asymptotics}
 \int_{\text{supp}\Psi} \frac {\Psi (y) d y} {y - x}\sim - \frac{1}{x} .
\end{equation}
Using the Zhukovsky transform from the Proposition \ref{Plemelj method} we also have that 

\begin{align}\label{asymptotic equation}
\begin{split}
\frac{1}{a} &=\lim_{z\rightarrow\infty} z[2Aaz^{-1} + 2Ba^{3}(3z^{-1}+z^{-3})]\\
&= 2Aa + 6Ba^{3}.
\end{split}
\end{align}

\begin{prop}\label{second moment of limiting distribution}
	The  moments of the  limiting distribution are given by 
	\begin{equation}\label{moments}
	m_{k} = \begin{cases}
	0 & \text{if $k$ is odd} \\
	C_{k/2}\left(\frac{6k}{k+4} Ba^{k+4} + a^{k}\right) & \text{if $k$ is even } 
	\end{cases}
	\end{equation}
	where $C_{k}$ denotes the $k$-th Catalan number.
\end{prop}	
\begin{proof}
	Odd moments are zero by symmetry. Now recall that the semicircle law's moments are  
	
	\begin{align*}
	\int_{-2a}^{2a} x^{2k}\sqrt{4a^{2}-x^{2}}dx = \frac{\pi (2a)^{2}}{2}\left(\frac{2a}{2}\right)^{2k} C_{k} = 2\pi a^{2(k+1)} C_{k}.
	\end{align*}
	Using this formula and equation (\ref{asymptotic equation}), we find
	\begin{align*}
	m_{2l}&= \int_{-2a}^{2a}x^{2l}\Psi(x)dx = \frac{1}{\pi}\int_{-2a}^{2a}x^{2l}(A +2Ba^{2} + Bx^{2}\sqrt{4a^{2}-x^{2}})dx\\
	&= \frac{(A +2Ba^{2})}{\pi}\int_{-2a}^{2a} x^{2l}\sqrt{4a^{2}-x^{2}}dx+ \frac{B}{\pi}\int_{-2a}^{2a}x^{2(l+1)}\sqrt{4a^{2}-x^{2}}dx\\
	&=2C_{l}(A+2Ba^{2})a^{2(l+1)} + 2C_{l+1}Ba^{2(l+2)}\\
	&= 2B(C_{l+1}-C_{l})a^{2(l+2)} + C_{l}a^{2l}.
	\end{align*}
	
	Lastly, one uses the relation between consecutive Catalan numbers i.e $(4l+2)C_{l} = (l+2) C_{l+1}$, to arrive at formula \ref{moments}.
\end{proof}

 Alternatively, one can derive  the same results using residue calculations. Using (\ref{asymptotic equation})  we may write the spectral density as

\begin {equation}\label{spectral denisty in terms of a}
\Psi (x) =\frac {1} {\pi} \left(\frac{1-2Ba^{4}}{2a^{2}} + Bx ^ {2}\right)\sqrt {4 a ^ {2} - x ^ {2}}_{+} =\frac {1} {\pi} \left(\frac{1}{2 a^{2}}-Ba^{2} + Bx ^ {2}\right)\sqrt {4 a ^ {2} - x ^ {2}}_{+}.
\end {equation}

\subsection*{A.2 The two-cut case}
 Define

\begin{equation*}
p(z)= (z^{2}-a^{2})(z^{2}-b^{2}),
\end{equation*}
and consider the resolvent, sometimes called the one-point function

\begin{equation}\label{resolvent}
W(z) = \int_{\text{supp} \,\Psi} \frac{\Psi(s)}{s-z}ds.
\end{equation}
The moments of $\Psi$ can be computed  as
\begin{equation}\label{resolvent moments}
m_{k} = \int_{\text{supp} \,\Psi} x^{k}\Psi(x)dx =-\frac{1}{2\pi i}\oint z^{k} W(z)dz, 
\end{equation}
where the contour is some circle that encloses the support of $\Psi$. For convenience of calculations, we multiply  (\ref{resolvent}) by a factor of $\frac{1}{i\pi}$ giving us the Borel transform of $\Psi$:
\begin{equation*}
\mathcal{B}(z) = \frac{1}{i\pi}\int \frac{\Psi(s)}{s-z}ds.
\end{equation*}

Define $J'(x) $ to be equal to the right hand side of equation \eqref{saddle point eqn}. It is not hard to show that $\mathcal{B}(z)$ gives us a standard scalar Riemann-Hilbert Problem (see \cite{Deift} section 1.2 and 6.7) with a solution given by

\begin{equation}\label{B(z)}
\mathcal{B}(z)= \frac{\sqrt{p(z)}}{2\pi i}\int_{\text{supp}\,\Psi}\frac{\frac{i}{\pi }J'(s)}{\sqrt{p(s)}_{+}}\frac{ds}{s-z}.
\end{equation}

To find the values of $a$ and $b$ we will require a set of conditions that may be derived in the following manner. Expand $\mathcal{B}(z)$ as
\begin{equation}\label{Borel transform asympototics}
\mathcal{B}(z) = \frac{\sqrt{p(z)}}{2\pi i}\left(-\frac{1}{z}\right)\int_{\text{supp} \,\Psi} \frac{\frac{i}{\pi}J'(s)}{\sqrt{p(s)}_{+}} \sum_{n=0}^{\infty}\left(\frac{s}{z}\right)^n ds.
\end{equation}
From this we can see that in order for $\mathcal{B}(z)\rightarrow 0$ as $z\rightarrow \infty$ two conditions on the moments must be satisfied:

\begin{equation}\label{moment conditions}
\int_{\text{supp}\,\Psi} \frac{\frac{i}{\pi}J'(s)}{\sqrt{q(s)}_{+}}s^{j}ds=0, \qquad  \text{for} \, \,  j =0 \, \,  \text{or}\, \,  j=1.
\end{equation}
Combining the limiting behavior with (\ref{Borel transform asympototics}), we achieve a third condition
\begin{equation}\label{normalization condition}
\frac{i}{2\pi}\int_{\text{supp}\,\Psi} \frac{J'(s)}{\sqrt{q(s)}_{+}}s^{2}ds= 1.
\end{equation}
We will compute these integrals below, giving us the conditions that we desire. The integral (\ref{B(z)}) can be computed using  the residue theorem.
Consider a dumbbell contour $\Gamma _{\varepsilon}$ around each cut. Then we obtain  
\begin{equation}\label{dumbbell contour}
2 \int_{\text{supp}\,\Psi}\frac{\frac{i}{\pi}J'(s)}{\sqrt{p(s)}}_{+}\frac{ds}{s-z} =\lim_{\varepsilon \rightarrow 0}\int_{\Gamma_{\varepsilon}} \frac{\frac{i}{\pi}J'(s)}{\sqrt{p(s)}}\frac{ds}{s-z} .
\end{equation}
 We can now apply  the residue theorem and (\ref{dumbbell contour}) to \eqref{B(z)} where the exterior domain is the outside of the dumbbell contour
\begin{align*}
\begin{split}
\mathcal{B}(z) &=  \sqrt{(z^{2}-a^{2})(z^{2}-b^{2})} \left[\frac{i J'(z)}{2\pi \sqrt{(z^{2}-a^{2})(z^{2}-b^{2})}} + \frac{1}{4\pi i}\oint _{\mathcal{C}} \frac{\frac{i}{\pi}J'(s)}{\sqrt{(s^{2}-a^{2})(s^{2}-b^{2})}}\frac{ds}{s-z}\right]\\
&= \frac{i}{2\pi} \left[J'(z) + \frac{\sqrt{(z^{2}-a^{2})(z^{2}-b^{2})}}{2\pi i}\oint _{\mathcal{C}} \frac{J'(s)}{\sqrt{(s^{2}-a^{2})(s^{2}-b^{2})}}\frac{ds}{s-z}\right],
\end{split}
\end{align*}	
where $\mathcal{C}$ is a  contour containing $z$ and $\text{supp}\,\Psi$. Using the binomial and geometric series
\begin{align*}
\begin{split}
\frac{s^{\alpha}}{\sqrt{(s^{2}-a^{2})(s^{2}-b^{2})}}\frac{1}{s-z} &=s^{\alpha} \left(1-\frac{a^{2}}{s^{2}}\right)^{-\frac{1}{2}}\left(1-\frac{b^{2}}{s^{2}}\right)^{-\frac{1}{2}}\frac{1}{s}\frac{1}{1-\frac{z}{s}}\\
&=s^{\alpha-3}\left(\sum_{k=0}^{\infty}{-\frac{1}{2}\choose k} (-1)^{k} \frac{a^{2k}}{s^{2k}}\right)\left(\sum_{j=0}^{\infty}{-\frac{1}{2}\choose j} (-1)^{j} \frac{b^{2j}}{s^{2j}}\right)\left(\sum_{q=0}^{\infty}\frac{z^{q}}{s^{q}} \right)\\
&= \sum_{k,j,q=0}^{\infty} {-\frac{1}{2}\choose k}{-\frac{1}{2}\choose j}  (-1)^{j+k} a^{2k} b^{2j}z^{q}s^{-2k-2j-q +\alpha-3},
\end{split}
\end{align*}
where $
{n \choose k} $
are the general binomial coefficients. To compute the residue we have to sum over all the coefficients of the terms with $s^{-1}$.  For example, for $\alpha = 7$ we have $(k,j,q) = (0,0,5), (1,0,3), (0,1,3), (1,1,1),\\ (0, 2, 1)$, and $(2, 0, 1).$ Thus
\begin{align*}
\begin{split}
\text{Res}[s^{7}/ \sqrt{q(s)},\infty] &= z^{5} + {-\frac{1}{2}\choose 1}(-1)a^{2}z^{3} + {-\frac{1}{2}\choose 1}(-1)b^{2}z^{3} + {-\frac{1}{2}\choose 1}{-\frac{1}{2}\choose 1}a^{2}b^{2}z + {-\frac{1}{2}\choose 2}b^{4}z + {-\frac{1}{2}\choose 2}a^{4}z\\
&= \left(\frac{1}{4}a^{2}b^{2}+ \frac{3}{8}(a^{4}+b^{4})\right)z+\frac{1}{2}(a^{2}+b^{2})z^{3}+z^{5}
\end{split}
\end{align*}
In table \ref{table 1} we have record all  the needed residues. 

\begin{table}
\begin{center}
	\begin{tabular}{ |c|c|c|c| }
		\hline
		$\alpha$ & Res$[s^{\alpha}/\sqrt{q(s)},\infty]$  \\
		\hline
		 1 & 0 \\ 
		 \hline
		 2 & 1 \\ 
		 \hline
		 3 & z \\
		 \hline
		 4 & $\frac{1}{2}(a^{2}+b^{2}) + z^{2}$ \\
		 \hline
		 5 & $\frac{1}{2}(a^{2}+b^{2})z + z^{3}$\\
		 \hline
		 6 & $\frac{3}{8}(a^{4}+b^{4}) + \frac{1}{4}a^{2}b^{2} + \frac{1}{2}(a^{2}+b^{2})z^{2} + z^{4}$ \\
		 \hline
		 7 & $\left(\frac{1}{4}a^{2}b^{2}+ \frac{3}{8}(a^{4}+b^{4})\right)z+\frac{1}{2}(a^{2}+b^{2})z^{3}+z^{5}$\\
		\hline
	\end{tabular}
\end{center}
\caption{\label{table 1}Residue Calculations}
\end{table}

Now write $J(s)$ compactly as 
\begin{equation*}
J(s) = \frac{A}{2}s^{2} + \frac{B}{4}s^{4}
\end{equation*}
for some real coefficients $A$ and $B$, where $B\not =0$. Then

\begin{align*}
\begin{split}
\text{Res}[(As +Bs^{3})/ \sqrt{q(s)},\infty] &=Bz.
\end{split}
\end{align*}
Thus explicitly 
\begin{equation}\label{resolvent solution}
\mathcal{B}(z) = \frac{i}{2\pi}\left[Az+ Bz^{3} - Bz\sqrt{(z^{2}-a^{2})(z^{2}-b^{2})}\right],
\end{equation}
so that by the Plemelj formula
\begin{equation}\label{Psi multicut}
\Psi (x) = \frac{B}{2\pi}|x|\sqrt{(x^2-a^2)(b^2-x^2)}_{+}.
\end{equation}

We now return to finding expressions for $a$ and $b$. Starting with the moment conditions
\begin{align*}
\begin{split}
0&= \int_{\text{supp}\,\Psi} \frac{J'(s)}{\sqrt{q(s)}_{+}}s^{j}ds= \oint \frac{As^{1+j} + Bs^{3+j}}{\sqrt{q(s)}}ds
\end{split}
\end{align*}
for $j=0,1$. The $j=1$ condition at $z=0$ gives us
\begin{align}\label{a in terms of b and g}
\begin{split}
\frac{-2A}{B} &= a^2 + b^2.
\end{split}
\end{align}
Next consider the normalization condition  (\ref{normalization condition}). Equivalently we may write it as
\begin{align}\label{normalization condition in practice}
\begin{split}
1&= \frac{i}{2\pi}\int_{\text{supp}\,\Psi} \frac{J'(s)}{\sqrt{q(s)}_{+}}s^{2}ds= \frac{1}{2\pi i}\oint \frac{As^{3} + Bs^{5}}{\sqrt{q(s)}}ds\\
&= \frac{A}{2}(a^{2}+b^{2}) + \frac{3B}{8}\left(a^{2}+ b^{2}\right)^{2} -\frac{B}{2}a^{2}b^{2}.\\
\end{split}
\end{align}
Equating conditions (\ref{a in terms of b and g}) and (\ref{normalization condition in practice}), we find that
\begin{equation}\label{a in terms of m2 and g}
a^{2} = \frac{-A +\sqrt{2B}}{B},
\end{equation}
and
\begin{equation}\label{b in terms of m2 and g}
b^{2} = \frac{-A -\sqrt{2B}}{B}.
\end{equation}
Lastly we wish to compute the second moment of our distribution. Using  equations (\ref{resolvent moments}) and (\ref{resolvent solution}) we find that 

\begin{equation*}
m_{2} =- \frac{1}{2}\oint \frac{i}{2\pi}\left[Az^{3}+ Bz^{5} - Bz^{3}\sqrt{(z^{2}-a^{2})(z^{2}-b^{2})}\right] dz =-\frac{B}{2}\frac{1}{2\pi i}\oint z^{3}\sqrt{(z^{2}-a^{2})(z^{2}-b^{2})}dz,
\end{equation*}
where the contour is a large circle that encloses $\text{supp}\Psi$. By Substiuting $z=1/w$ we can solve this contour integral by computing the residue at infinity:
\begin{align}\label{double cut second moment}
\begin{split}
&= -\frac{B}{2}\frac{1}{2\pi i}\oint \frac{1}{w^{3}}\sqrt{\left(\frac{1}{w^{2}}-a^{2}\right)\left(\frac{1}{w^{2}}-b^{2}\right)}\frac{1}{w^{2}}dw\\
&= -\frac{B}{2}\frac{1}{2\pi i}\oint \frac{1}{w^{7}}\sqrt{\left(1-a^{2}w^{2}\right)\left(1-b^{2}w^{2}\right)}dw.\\
\end{split}
\end{align}
Since
\begin{align*}
\begin{split}
\frac{1}{w^{7}}\sqrt{\left(1-a^{2}w^{2}\right)\left(1-b^{2}w^{2}\right)}dw &= \sum_{k=0}^{\infty}\sum_{j=0}^{\infty}{\frac{1}{2}\choose k}{\frac{1}{2}\choose j}(-1)^{k+j}a^{2k}b^{2j} w^{2k + 2j -7},
\end{split}
\end{align*}
the contributing terms are from when $k+j = 3$, that is $(k, j)= (3, 0)$, $(2, 1)$, $(1, 2)$, and $(0, 3)$. Thus
\begin{align}\label{second moment C2}
\begin{split}
m_{2} &=-\frac{B}{2}\left(-{\frac{1}{2}\choose 3}a^{6} - {\frac{1}{2}\choose 2}{\frac{1}{2}\choose 1}a^{4}b^{2}  -{\frac{1}{2}\choose 1}{\frac{1}{2}\choose 2}a^{2}b^{4} -{\frac{1}{2}\choose 3} b^{6}\right)\\
&= \frac{B}{32}(a^{6} + b^{6} -a^{4}b^{2}- a^{2}b^{4})\\
&= \frac{B}{32}(a^{2}-b^{2})^{2}(a^{2}+b^{2}).
\end{split}
\end{align}
Similarly we may also compute a new normalization condition 

\begin{align}\label{normalization new}
\begin{split}
1 &=-\frac{B}{2}\left({\frac{1}{2}\choose 2}a^{4} + {\frac{1}{2}\choose 1}{\frac{1}{2}\choose 1}a^{2}b^{2}  +{\frac{1}{2}\choose 2} b^{4}\right)\\
&= \frac{B}{16}(a^{2}-b^{2})^{2}.
\end{split}
\end{align}
Combining equations (\ref{second moment C2}), (\ref{normalization new}), (\ref{a in terms of m2 and g}), and (\ref{b in terms of m2 and g}),  we arrive at

\begin{align}\label{second moment isolated C2}
\begin{split}
m_{2}&= -A/B.
\end{split}
\end{align}


\begin{thebibliography}{9}
	\bibitem{Random Matrices}
	G. Anderson, A. Guionnet, and O. Zeitouni. \textit{An Introduction to Random Matrices}.  Cambridge: Cambridge University Press (Cambridge Studies in Advanced Mathematics). 2009.

	\bibitem{Shahab thesis}
	S. Azarfar. Topological Recursion and Random Finite Noncommutative Geometries. Ph.D. Thesis, University of Western Ontario, 2018.
	 
	\bibitem{AK}
	S. Azarfar and M. Khalkhali. Random Finite Noncommutative Geometries and Topological Recursion.
		arXiv:1906.09362. 
	
	
	
	\bibitem{Barrett2015} 
	J. W.  Barrett. Matrix Geometries and Fuzzy Spaces as Finite Spectral Triples. \textit{J. Math. Phys.},
	56(8):082301, 25, 2015.
	
	\bibitem{Barrett2016}
	J. W.  Barrett and L.  Glaser. Monte Carlo Simulations of Random Non-commutative Geometries.
	\textit{J. Phys. A}, 49(24):245001, 27, 2016.
	
	\bibitem{blobbed}
	G. Borot. Blobbed topological recursion. \textit{Theor. Math. Phys.}, 185(3):1729 -1740,
	2015.
	
	
	\bibitem{Brezin1978}
	E. Brezin, C. Itzykson, G. Parisi and J.B. Zuber. Planar Diagrams. \textit{Commun. Math. Phys.} 59, 35-51, 1978.
	
	
	\bibitem{Connes94}
 A. Connes.  {Noncommutative geometry}. Academic Press, Inc., San Diego, CA, 1994.
	
	\bibitem{Connes95}
A.~Connes.
\newblock Geometry from the spectral point of view.
\newblock {\em Lett. Math. Phys.}, 34(3):203-238, 1995.

\bibitem{Connes-2013}
A.~Connes.
\newblock On the spectral characterization of manifolds.
\newblock {\em J. Noncommut. Geom.}, 7(1):1-82, 2013.
	

\bibitem{Connes-VS}
A. Connes and  W.  D. van Suijlekom. Spectral truncations in noncommutative geometry and operator systems. Commun. Math. Phys. 2020. 	arXiv:2004.14115.
	

	
	\bibitem{Deift}	
	P. Deift. \textit{Orthogonal polynomials and random matrices: A Riemann-Hilbert approach}.  Courant lecture notes 3 New York: New York University. 1999.
	
	
	
	
		
	\bibitem{Eynard2018}
	B. Eynard. \textit{ Counting Surfaces}, volume 70 of \textit{Progress in Mathematical Physics}. Birkh\"{a}user/Springer,  CRM Aisenstadt chair lectures. 2016.
	
	\bibitem{planar II}
C. 	Itzykson and   J. B. Zuber. The planar approximation. II. \textit{J. Math. Phys.} 21, 1980, no. 3, 411-421.
	
	\bibitem{Kazakov 1}
	V.A. Kazakov, M. Staudacher and T. Wynter. Exact Solution of Discrete Two-
	Dimensional R2 Gravity. \textit{Nucl. Phys. B} 471, 309-333, 1996.

	\bibitem{Kazakov 2}
	V.A. Kazakov, M. Staudacher and T. Wynter. Almost flat planar graphs. \textit{Comm. Math.Phys.} 179, 235-256, 1996.
	
	\bibitem{Migdal}
	A. A. Migdal. Recursion equations in lattice gauge theories. \textit{Sov. Phys. JETP} 42 (1975)
	413, Zh.Eksp.Teor.Fiz. 69, 810-822, 1975.

		\bibitem{Universality}
	L. Pastur  and M. Shcherbina. Universality of the local eigenvalue statistics for a class of unitary invariant random matrix ensembles. \textit{J Stat Phys} 86, 109-147, 1997. 
	
	\bibitem{Sanchez}
	C. Perez-Sanchez. Computing the spectral action for fuzzy geometries: from random noncommutatative geometry to bi-tracial multimatrix models.  	arXiv:1912.13288. 
			
\end{thebibliography}
\end{document}